\documentclass[11pt]{article}
\usepackage[british]{babel}
\usepackage{geometry}
\usepackage{graphicx}
\usepackage{amsmath}
\usepackage{amssymb}
\usepackage{amsthm}
\usepackage{enumerate}
\newtheorem{thm}{Theorem}
\newtheorem{lm}{Lemma}
\newcommand{\GF}{\mathrm{GF}}
\newcommand{\NN}{\mathbb{N}}
\newcommand{\FF}{\mathbb{F}}
\newcommand{\cC}{\mathcal{C}}
\newcommand{\cP}{\mathcal{P}}
\newcommand{\cR}{\mathcal{R}}
\newcommand{\cL}{\mathcal{L}}
\newcommand{\cO}{\mathcal{O}}
\newcommand{\cS}{\mathcal{S}}
\newcommand{\PG}{\mathrm{PG}}
\newcommand{\PGL}{\mathrm{PGL}}
\newcommand{\bPG}{\mathbb{PG}}

\title{LDPC codes from Singer cycles}
\author{ \\
Luca Giuzzi \\[-7pt]
\\
Dipartimento di Matematica \\
Politecnico di Bari \\
Via Orabona, 4 \\
70125 Bari, Italy \\
\emph{Email:} \texttt{giuzzi@poliba.it}
\and
\\
Angelo Sonnino \\[-7pt]
\\
Dipartimento di Matematica e Informatica \\
Universit\`a della Basilicata \\
Campus Macchia Romana \\
Viale dell'Ateneo Lucano, 10 \\
85100 Potenza, Italy \\
\emph{Email:} \texttt{angelo.sonnino@unibas.it}}
\begin{document}
\maketitle
\begin{abstract}
The main goal of coding theory is to devise efficient systems to
exploit the full capacity of a communication channel, thus achieving
an arbitrarily small error probability. Low Density Parity Check
(LDPC) codes are a family of block codes---characterised by admitting
a sparse parity check matrix---with good correction capabilities. In
the present paper the orbits of subspaces of a finite projective space
under the action of a Singer cycle are investigated.
\end{abstract}
\section{Introduction}
A $[n,k,d]$--linear code over $\GF(q)$ is a monomorphism
$\theta$ from
$M=\GF(q)^k$ into $R=\GF(q)^n$ such that the images of any two
distinct vectors $\mathbf{m_1},\mathbf{m_2}\in M$ differ in 
at least $d$ positions. The elements of $M$ are called
messages, while the elements of the image $\cC=\theta(M)$ are
the codewords of $\theta$. The function
\[ d:\begin{cases}
  R\times R\mapsto\NN \\
  (\mathbf{x},\mathbf{y})\mapsto |\{i: x_i-y_i\neq 0\}| 
  \end{cases}\]
is the \emph{Hamming distance on $R$}.
In the present paper we shall usually identify a code
with the set of its codewords.

The problem of \emph{minimum distance decoding} is to
find, for any given vector $\mathbf{r}\in R$ the
set $\cC_{\mathbf{r}}$ of all the codewords
$\mathbf{c}\in\cC$ at minimum Hamming distance
from $\mathbf{r}$.
If $\cC_{\mathbf{r}}$ contains just one element $\mathbf{c}$,
then we can uniquely determine a message $\mathbf{m}$ such
that $\theta(\mathbf{m})=\mathbf{c}$ and we
state that the decoding of $\mathbf{r}$
has succeeded; otherwise, we remark that it has not been
possible to correctly recover the message originally sent.
\par
Minimum distance decoding is, in general, a hard 
problem, see \cite{NP1}; in fact, many algorithms currently
in use  sacrifice some of the abstract correcting capabilities
of a code in favour of ease of implementation and lower
complexity; notable examples
are the syndrome decoding technique for general linear
codes and the Welch--Berlekamp approach for
BCH codes, see \cite{Me}.
We remark, however, that even these techniques may be
prohibitively expensive when the length $n$ of
the code, that is the dimension of the vector
space $R$, is large.
\par
On the other hand, long codes present several
advantages, since it can be shown that almost
all codes with large $n$ have excellent correction
capabilities, see \cite{WC}.
It is thus important to
find some special families of codes for which good
encoding and decoding
techniques are known.
\par
Low Density Parity Check (LDPC) codes have been
introduced by Gallager in \cite{Gal1},\cite{Gal2} and then, ignored
for almost 30 years.
They have been recently rediscovered, see
\cite{M1}; and it has been realised that
they may be applied to high--speed, high--bandwidth
digital channels and support efficient decoding algorithms
based upon message--passing strategies.
Furthermore,
the performance of some of these codes is remarkably
close to the Shannon limit for the AWGN channel;
consequently, they turn out to be very competitive,
even when compared with more elaborate constructions,
like  turbo codes, see \cite{M1}.
Nevertheless, the problem of providing efficient encoding is,
in general,  non--trivial, see \cite{NP2}, although, in several cases,
manageable, see  \cite{ru}.
This motivates the search for new ways of
constructing suitable parity--check matrices for broad classes of LDPC
codes.
\par
A linear code is Low Density Parity Check (for short, LDPC) if
it admits at least one sparse parity check matrix.
In particular, a LDPC code is
\emph{regular} if the 
set of its codewords is
the null space of a parity check matrix $\mathbf{H}$
with the following structural properties:
\begin{enumerate}[(L1)]
\item each row of $\mathbf{H}$ contains $k$ non--zero entries;
\item each column of $\mathbf{H}$ contains $r$ non--zero entries;
\item the number of non--zero entries in common between any two
distinct columns of $\mathbf{H}$ is at most $1$;
\item both $k$ and $r$ are small compared with the length $n$ of the
code and the number of rows in $\mathbf{H}$.
\end{enumerate}
In this paper we describe an algorithm for the construction of
such parity--check matrices  based on the orbits of subspaces
of a finite projective space $\mathrm{PG}(n-1,q)$, with $q$ even,
under the action of a Singer cycle. More precisely, some cyclic and
almost cyclic LDPC codes are constructed using some suitable
representatives for each of these orbits.

\section{Preliminaries: incidence matrices} 
There is a
straightforward correspondence between
finite incidence structures and binary matrices.
This topic has been widely investigated, also
in the context of coding theory; see \cite{AK}.

Given a binary matrix $M$, it is always possible to
introduce an
incidence structure $\cS_M=(\cP,\cL,I)$ such that
the points $\cP$ are the columns of $M$, the blocks
$\cL$ are the rows of $M$ and $P\in\cP$ is
incident with $L\in\cL$ if and only 
if $m_{LP}=1$. Conversely,
given an
incidence structure $\cS$ with
point--set $\cP=\{p_1,p_2,\ldots,p_v\}$
and block--set $\cL=\{l_1,l_2,\ldots,l_b\}$,
the \emph{incidence matrix} $M=(m_{ij})$ of $\cS$
is the binary $b\times v$--matrix with
\[
m_{ij}=\begin{cases}
  1 & \mbox{ if $p_jIl_i$} \\
  0 & \mbox{ otherwise.}
\end{cases}
\]
Recall that an incidence structure
$\cS=(\cP,\cL,I)$ is written simply
as $(\cP,\cL)$ when any element $L\in\cL$
is a subset of $\cP$, and given $p\in\cP$
and $L\in\cL$, we have $pIL$ if, and only if,
$p\in L$.

A map
 $\varphi:\cP\cup\cL\rightarrow\cP\cup\cL$ is
a collineation of $\cS$ when $\varphi$ maps points into points, blocks
into blocks and preserves all incidences.

In particular, we are interested in incidence structures
endowed with
%
collineation group $G$ acting regularly on the points.
In this
case it is quite easy to write all the blocks in
$\cL$ and the associated incidence matrix $\mathbf{H}$
has a special form.
We proceed as follows.

Fix a point $P\in\mathcal{P}$ and let
$\mathcal{T}=\{\ell_{1},\ell_{2},\ldots ,\ell_{h}\}$ be the set of all
blocks of $\mathcal{S}$ incident with $P$ and such that:
\begin{enumerate}[1.]
\item $\ell_{i}^{g}\neq\ell_{j}$ for any $g\in G$ and $1\leq i<j\leq
h$;
\item for any line $\ell\in\mathcal{L}$ there is a block
$\ell_{i}\in\mathcal{T}$ and a $g\in G$ such that $\ell
=\ell_{i}^{g}$.
\end{enumerate}
The set $\mathcal{T}$ is called a \emph{starter set} for 
$\cS$, see
\cite{BJL}. Clearly,
 \[ \bigcup_{j=1}^{h}\{\,\ell_{j}^{g}\mid g\in G\,\}; \]
 therefore, the whole incidence matrix $\mathbf{H}$ of
 $\cS$ can be reconstructed by just providing a suitable
 starter set $\mathcal{T}$ and generators of the group $G$.

If we further suppose $G$ to be cyclic, let $\tau$ be one of its
generators, then the incidence structure
$$\mathcal{S}=(\mathcal{P},\{\,\ell_{j}^{\tau^{i}}\mid1\leq i\leq
|G|,\ 1\leq j\leq h\,\})$$
admits at least a
circulant incidence matrix $\mathbf{H}$; that is,
a block matrix $\mathbf{H}$ 
of type
$$\mathbf{H}=\begin{pmatrix}
H_{1} \\ \vdots \\ H_{h}
\end{pmatrix}$$
wherein any row $H_{j}$, $j>1$ is obtained from the preceding
one $H_{j-1}$ by applying a cyclic right shift.
Clearly, this happens provided that the
points and blocks of $\mathcal{S}$ are arranged is such a way as
$P_{i}=P^{\tau^{i-1}}$ and $\ell_{i}=\ell_{j}^{\tau^{i-1}}$ with
$j\in\{1,2,\ldots ,h\}$.
\par


\section{Preliminaries: projective spaces and spreads}
Let $\PG(V,\FF)$ be the projective space whose elements are the vector
subspaces of the vector space $V$ over the field $\FF$..
 We denote by the
same symbol a $1$--dimensional vector subspace of $V$ and the
corresponding element of $\PG(V, \FF)$.
An element $T$ of $\PG(V,\FF )$ has \emph{rank} $t$ and
\emph{dimension} $t-1$, whenever $T$ has dimension $t$ as a vector space over
$\FF$.
When the dimension of $V$ over $\FF =\GF(q)$ is finite and equal to
$n$, we shall usually write $\PG(n-1,q)$ instead of $\PG(V,\FF )$.
The elements
of rank $1$, $2$, $3$ and $n-1$ in $\PG(n-1,q)$ are  called
respectively \emph{points},
\emph{lines}, \emph{planes} and \emph{hyperplanes}.
 Points contained
in the same line are said to be \emph{collinear}.
Observe that, for any $i\geq 1$, 
\[ \PG_i(V)=(\cP,\cL), \]
where $\cP=\{W\leq V: \dim W=1\}$ and
$\cL=\{X\leq V: \dim X=i+1\}$ is an incidence 
structure.

Let now
$\{E_0,E_1,\dots,E_{n-1}\}$ be a fixed basis of $V$.
Then, we say that the point
$\langle x_0E_0+x_1E_1,+\cdots+x_{n-1}E_{n-1}\rangle$ of $\PG(V,\FF )$
has homogeneous projective coordinates $(x_0,x_1,\dots,x_{n-1})$.

The number of points of $\PG(n-1,q)$ is 
$$\frac{q^{n}-1}{q-1}=q^{n-1}+\cdots+q+1,$$ 
while the number of its lines is 
$$\frac{(q^{n}-1)(q^{n-1}-1)}{(q^2-1)(q-1)},$$
see \cite{PGOFF}.
In particular, the number of lines of $\PG(2,q)$,
$\PG(3,q)$ and $\PG(4,q)$ are
respectively  $q^2+q+1$, $(q^2+1)(q^2+q+1)$  and 
$(q^2+1)(q^4+q^3+q^2+q+1)$.

The points and the lines of a projective space $\PG(n-1,q)$
form a $2$--design, whose incidence matrix $M$
defines a regular LDPC code, called  $\bPG^{(1)}$
in \cite{LDPC}.
The code defined by the transposed matrix $M$ is also
LDPC and it is called $\bPG^{(2)}$ in the aforementioned paper.
%


An application $A$ of $V$ in itself is a semilinear map 
if, and only if, there is an automorphism $\mu$ of $\FF$ such that for
all vectors $v,w\in V$ and all elements $\alpha\in\FF$,
\[(v+w)^A=v^A+w^A,\qquad (\alpha v)^A=\alpha^{\mu}v^A.\]
If $A$ is a bijection, then we say that it is \emph{non--singular}.
When $\mu=id$, then $A$ is called \emph{linear}. 
Any non--singular semilinear  map of $V$ in itself
induces a collineation $\tau_A$ of $\PG(V,\FF )$
which maps the point $\langle v \rangle$ into the point
$\langle v^A\rangle$.
Conversely, given any collineation $\tau$ of $\PG(V,\FF )$,
there is a non--singular semilinear map $A$ of $V$ such that $\tau=\tau_A$. 

A Singer cycle $S$ is a cyclic collineation group of $\PG(n-1,q)$
acting regularly on the points; that is to say that $S$ has order
$q^{n-1}+\cdots+q+1$,  is cyclic and only the identity fixes any
point.
If $S$ and $S'$ are two Singer cycles of $\PG(n-1,q)$,
then there is a collineation $\tau$ of  $\PG(n-1,q)$
such that $S'=\tau^{-1} S \tau$. 
    
As usual, we write $\PG(n-1,q)$ for $\PG(\GF(q^{n}),\GF(q))$.
Let $\alpha$ be a generator of the multiplicative group of $\GF(q^n)$;
then, the map $\sigma$ of $\GF(q^{n})$
into itself defined by $\sigma :x \mapsto \alpha x$
is a non--singular $\GF(q)$--linear map.
Observe that $\sigma$, as a linear map, has order $q^{n}-1$
and defines a collineation of $\PG(n-1,q)$ of order
$q^{n-1}+\cdots+q+1$ 
acting transitively on the points of $\PG(n-1,q)$; hence,
the collineation group $S$ generated by $\sigma$ is indeed
a Singer cycle of $\PG(n-1,q)$. 

We can regard the projective space $\PG(n-1,q)=\PG(V,\GF(q))$
as a distinguished hyperplane of $\PG(n,q)=\PG(V',\GF(q))$,
with $V'=\langle E \rangle \oplus V$.
Let now $\sigma$ be the generator of a Singer cycle of $\PG(n-1,q)$.
 The map 
 \[ \sigma' : x E + v \mapsto x E + v^{\sigma} \]
 gives a cyclic collineation group $\tilde{S}$ of order $q^{n}-1$,
 which induces a Singer cycle on the hyperplane $\PG(n-1,q)$ and acts
 regularly on the points of $\PG(n,q)\setminus\PG(n-1,q)$
 different from $\langle E\rangle$.
 Furthermore, this group $\tilde{S}$ acts transitively on
 the lines of $\PG(n,q)$ incident with $\langle E \rangle$.

A \emph{$(t-1)$--spread} $\cS$ of a projective space $\PG(n-1,q)$
is a family  of mutually disjoint subspaces, each of rank $t$,
such that each point of $\PG(n-1,q)$ belongs to exactly
one element of ${\cS}$.
It has been proved by Segre \cite{Segre1964} that
a $(t-1)$--spread of $\PG(n-1,q)$ exists if and only if $n=rt$.
A spread   ${\cS}$ with $t=2$ is called  \emph{line--spread}.

Let now $\cS$ be a $(t-1)$--spread of $\PG(rt-1,q)$.
It is possible to embed $\PG(rt-1,q)$ into $\PG(rt,q)$
as a hyperplane, and then to introduce a new incidence
structure $A({\cS})$  as follows.
The points of $A({\cS})$  are the points of
$\PG(rt,q)\setminus\PG(rt-1,q)$.
The lines of $A({\cS})$ are the $t$--dimensional subspaces
of $\PG(rt,q)$ which are not contained in $\PG(rt-1,q)$
but contain an element of ${\cS}$.
The incidence relation is the natural set-theoretical one.
The incidence structure $A({\cS})$ is a $2-(q^{rt},q^t,1)$
translation design with parallelism, see \cite{Barlotti-Cofman1974}. 
The $(t-1)$--spread $\cS$  is called Desarguesian when
 $A({\cS})$ is isomorphic to the affine space $AG(r,q^t)$.
We provide two different characterisations of Desarguesian spreads,
according as $r=2$ or $r\neq 2$.

When $r=2$,  the projective space has dimension $2t-1$.
A \emph{regulus} ${\cR}$ of $\PG(2t-1,q)$
is a set of $q+1$ mutually disjoint $(t-1)$--dimensional
subspaces such that each line intersecting three elements
of ${\cR}$ has a point in common with all the subspaces of ${\cR}$. 

If $A$, $B$, $C$ are three mutually disjoint 
$(t-1)$--dimensional subspaces of $\PG(2t-1,q)$,
then there is a unique regulus ${\cR}(A,B,C)$ of $\PG(2t-1,q)$
containing $A$, $B$ and $C$.
A spread ${\cS}$ is \emph{regular} if the regulus ${\cR}(A,B,C)$
is contained in ${\cS}$ whenever $A$, $B$ and $C$ are three
distinct element of ${\cS}$. 
 
\begin{thm}[\cite{Bruck-Bose1964}]
Suppose $q>2$. A $(t-1)$--spread $\cS$ of $\PG(2t-1,q)$
is Desarguesian if and only if it is regular.
\end{thm}

We now consider the case $r>2$.
A $(t-1)$--spread is \emph{normal} when it induces a spread in
any subspace generated by any two of its elements, see
\cite{Lunardon1997}.
In particular, fix $T=\langle A,B\rangle$ with $A,B\in\cS$.
Then, for any $C\in\cS$, either $C\subseteq T$ or $C\cap\cS=\emptyset$.
Such spreads are called \emph{geometric} in \cite{Segre1964}.
 
\begin{thm}[\cite{Barlotti-Cofman1974}] For $r >2$,  the 
$(t-1)$--spread $\cS$ is Desarguesian if and 
only if it is normal.
\end{thm}

\section{The $\GF(q)$--linear representation of  $\PG(r-1,q^t)$}

Let $V$ be a $r$--dimensional vector space over $\GF(q^t)$,
and let, as usual, $\PG(r-1,q^t)=\PG(V,\GF(q^t))$.

We may regard $V$ as
a vector space of dimension $rt$ over $\GF(q)$;
hence, each point $\langle x\rangle$ of
$\PG(r-1,q^t)$ determines a $(t-1)$--dimensional subspace $P(x)$ of the
projective space $\PG(V,\GF(q))=\PG(rt-1,q)$; likewise,  each line $l$
of $\PG(r-1,q^t)$ defines a $(2t-1)-$dimensional subspace $P(l)$ of
$\PG(rt-1,q)$.

Write ${\cL}$ for  the set of all
$(t-1)$--dimensional subspaces of
$\PG(V,\GF(q))$, each obtained as $P(x)$, with
 $x$   a point of $\PG(r-1,q^t)$.
Then, ${\cL}$ is a $(t-1)$--spread of
$\PG(rt-1,q)$; this is called the $\GF(q)$--linear
representation of $\PG(r-1,q^t)$. 
It has been shown that ${\cL}$ is Desarguesian
and any Desarguesian spread of $\PG(rt-1,q)$ is isomorphic to ${\cL}$,
see \cite{Bruck-Bose1964} for $r=2$ and
\cite{Segre1964}, \cite{Barlotti-Cofman1974} for $r>2$. 


\begin{thm}
\label{thm:3}
A $(t-1)$--spread $\cS$  of $\PG(rt-1,q)$ is Desarguesian if and only if
there is a collineation group of $\PG(rt-1,q)$ of order
$q^{t-1}+q^{t-2}+\cdots+q+1$ fixing all elements of $\cS$.
\end{thm}
\begin{proof}
Denote by $T$ the translation group of $A({\cS})$,
that is, the group of all elations of $\PG(rt,q)$ with axis $\PG(rt-1,q)$.
Let $O$ be a fixed point of $A({\cS})$.
For each line $L$ of $A({\cS})$ incident with $O$, denote by
$T_L$ the stabiliser of $L$ in $T$; take also
${\cal K}$ to be the family of all the subgroups $T_L$ of $T$.
We know that $T$ is elementary abelian and $T_L$ is transitive on the
points of the line $L$. 
We may now
introduce a new incidence structure $\pi$,
whose points are the elements of $T$ and whose lines
are the lateral classes of the subgroups $T_L$.
Given a point $P\in A({\cS})$, denote by $\tau_{O,P}$ the element
of $T$ which maps $O$ into $P$.
The map $P \mapsto \tau_{O,P}$ turns out to be an isomorphism
between $A({\cS})$ and $\pi$.

The \emph{kernel} $K$ of $\cal K$
is the set of all the endomorphisms $\alpha$ of $T$
such that $T_L^{\alpha}\subset T_L$.
It has been shown
in \cite{Bonetti-Lunardon}  that $K$ is a field.
Hence, $T$ is a vector space over $K$ and each
element of ${\cal K}$ is a vector subspace of $T$.
Given any central collineation $\omega$ of $\PG(rt,q)$
with axis $\PG(rt-1,q)$ and centre $O$,
the map $\bar{\omega}$ of $T$ into itself defined
 by $\bar{\omega}:\tau \mapsto \omega \tau \omega$ is
an element of $K$.
Hence, $K$ contains a subfield isomorphic to $\FF=\GF(q)$. 

Let now $E$ be a subfield of $K$ and denote by $\PG(T,E)$
the projective space associated to $T$, regarded
as a vector space over $E$,
and by ${\cal K}(E)$ the spread of $\PG(T,E)$ induced by $\cal K$.
The designs $\pi$ and $A({\cal K}(E))$ are isomorphic.
Furthermore, the $(t-1)$--spreads ${\cS}$ and ${\cal K}(F)$
are also isomorphic, that is, there is a collineation
$\tau$ of $\PG(rt-1,q)$ such that ${\cS}^{\tau}={\cal K}(F)$,
see \cite{Bonetti-Lunardon}. 
 
It follows that
the spread $\cS$ is Desarguesian if and only if
$T_L$ has dimension $1$ over $K$, see \cite{Bonetti-Lunardon}.
This condition is equivalent to require that $K$ has order $q^{t}-1$
and defines a collineation group of $\PG(rt-1,q)$
of order $q^{t-1}+q^{t-2}+\cdots+q+1$ fixing all
the elements of  ${\cal K}(F)$.
 \end{proof}

\begin{thm}
\label{thm:4}
Let $S$ be a Singer cycle of $\PG(n-1,q)$, with $n=rt$.
Denote  by  $S_1$ and $S_2$ respectively  the subgroups of $S$ of order
$\frac{q^t-1}{q-1}$ and  $\frac{q^n-1}{q^t-1}$, so that
$S=S_1 \times S_2$.
Then, there is a Desarguesian $(t-1)$--spread $\cS$ of $\PG(n-1,q)$
such that $S_2$ acts regularly on $\cS$ and $S_1$ fixes all its elements.
\end{thm}
\begin{proof}
Write $\PG(n-1,q)=\PG(\GF(q^n),\GF(q))$ and assume $S$ to
be the Singer cycle spanned by the collineation
$\sigma :x \mapsto \alpha x$, where $\alpha$
is a generator of the multiplicative group of $\GF(q^n)$. 

As $t$ divides $n$, the element $\beta=\alpha^{\frac{q^n-1}{q^r-1}}$
is a generator of the multiplicative group of the subfield $\GF(q^t)$
of $\GF(q^n)$.
Now let $S_1$ be the subgroup of $S$ generated by
$\sigma_1=\sigma^{\frac{q^n-1}{q^r-1}}$. 
Then, ${\cS}=\{\,\GF(q^t)x \mid x \in\GF(q^n)\,\}$ is a
$(t-1)$--spread of $\PG(n-1,q)$ which is preserved
by $S_1$.
As $S_1$ acts over each member of $\cS$ as a Singer cycle,
by Theorem \ref{thm:3} we conclude that $\cS$ is Desarguesian.

If $\gamma$ is a primitive element
 of $\GF(q^n)$ over $\GF(q^t)$, then the collineation defined by
 the map $\sigma_2: x \mapsto \gamma x$ defines a subgroup
 $S_2$ of $S$ of order $\frac{q^n-1}{q^t-1}$.
By construction, $S =S_1\times S_2$.
 As
 ${\cS}=\{\,\GF(q^t)\gamma^j \mid (q^t)^{r-1}+\cdots+q^t+1\geq j\geq
 0\,\}$,
 the group $S_2$ preserves the $(t-1)$--spread
 $\cS$ and acts regularly on its elements.
 \end{proof}

\section{Decomposition of $\PG(n-1,q)$ for $n$ odd}

For odd $n$ any line of $\PG(n-1,q)$ has an orbit of
length $\frac{q^n-1}{q-1}$ under the action of a Singer cycle $S$ of
the space.
Hence, we can decompose the set of all lines of $\PG(n-1,q)$ into
$\frac{q^{n-1}-1}{q^2-1}$ orbits under the action of $S$.
Each of these orbits, say $i$ for $1\leq i\leq
\frac{q^{n-1}-1}{q^2-1}$,  defines a cyclic structure whose
incidence matrix $M_i$  is circulant.
Hence, the incidence matrix $M$ of $\bPG^{(1)}$ has the following structure:
$$M=\left(
  \begin{array}{c}
    M_1\\
    M_2\\
    \vdots \\
    M_{(q^{n-1}-1)/(q^2-1)}
  \end{array}\right).$$ 

A starter set of $\PG(n-1,q)$ can be obtained as follows.
Let $\sigma$ be a generator of the Singer cycle $S$ and choose
a point $P$. Fix a line $l$ incident with $P$,
suppose that $i_0=0, i_1, \dots,i_{q}$ are integers,
 and
 $P=P^{\sigma^{i_0}}, P_1=P^{\sigma^{i_1}}, \dots,
 P_q=P^{\sigma^{i_q}}$
 are the points of $l$.
 Then, $l_j=l^{\sigma^{i_j}}$, $0\leq j\leq q$
 are exactly the $q+1$ lines of the orbit of $l$
 under the action of $S$ which are incident with $l$. 
Hence, a starter set of $\PG(n-1,q)$ is just 
 ${\bf S}=\{s_1, s_2,\dots, s_{\frac{q^{n-1}-1}{q^2-1}}\}$,
 consisting of $\frac{q^{n-1}-1}{q^2-1}$ lines incident
 with $P$ such that, if $P^{\sigma^h}$ belongs to $s_i$,
 then $s_i^{\sigma^h}$ does not belong to ${\bf S}$.

\section{Decomposition of $\PG(n-1,q)$ for $n$ even}

Suppose now  $n=2t$ with $t>1$ and denote by $S$ a Singer cycle
of $\PG(n-1,q)$. 


Write $S=S_1\times S_2$,
where
$S_1$ has order $\frac{q^2-1}{q-1}$ and $S_2$ has
order $\frac{q^{2t}-1}{q^2+1}$.
By Theorem \ref{thm:4}, there is a Desarguesian line
spread $\cS$ of $\PG(2t-1,q)$
such that $S_1$ fixes all the lines of $\cS$ and $S_2$
acts regularly its elements.
 

\begin{lm}
The stabiliser in $S$ of a line $m$ not in $\cS$ is the identity. 
\end{lm}
\begin{proof}
It has been proved in \cite{TXKLA}, that there are exactly
$\frac{q^{n}-1}{q^2-1}$ lines of $\PG(n-1,q)$ whose stabiliser in $S$
is different from the identity. 
Since any line of $\cS$ is fixed by $S_1$ and $\cS$ contains exactly
$\frac{q^{n}-1}{q^2-1}$ of them,
the stabiliser in $S$ of a line $m$ not in $\cS$ is the identity.
\end{proof}

The set of all the lines of
$\PG(n-1,q)$ can be decomposed into the set $\cS$ and
$q(q^{2t-4}+ q^{2t-6}+\cdots+ q^2+1)$ other
orbits under the action of $S$, each of length
$\frac{q^{2t}-1}{q-1}$.
These  orbits, say $i$, for $1\leq i\leq q(q^{2t-4}+ \cdots+ q^2+1)$
  define a cyclic structure whose
  incidence matrix $M_i$ is  circulant.
  Hence, the incidence matrix $M$ of
  $\bPG^{(1)}$ has in this case the following structure
  $$M=\left(\begin{array}{c}
      M_0\\
      M_1\\
      \vdots \\
      M_{q(q^{2t-2}-1)/(q^2-1)}
    \end{array}\right),$$
  where $M_0$ is the incidence matrix of the 
  structure induced on $\cS$.
  The points of $\PG(n-1,q)$ may be indexed
  in such a way as to
  have
  $$M_0=\left(\begin{array}{cccc}
      B_1 & B_2 &  \cdots & B_{q+1}
    \end{array}\right),$$
  where $B_1=B_2= \cdots=B_{q+1}$ is the identity matrix of
  order $q^{n-1}+q^{n-2}+\cdots +q+1$.






\section{The case of  $\PG(3,2^e)$ }
In order to investigate the details of what happens in
$\PG(3,2^{e})$ we need to recall some properties of
elliptic quadrics and regular spreads in this space.
The interested reader might look at \cite{H} for a proof of
the results.

Denote by $Q^-(3,q)$ the set of all points of $\PG(3,q)$ whose
homogeneous coordinates are solution of the equation
\[ X_0X_1+X_2^2+bX_2X_3+ cX_3^2=0, \]
where $b$ and $c$ are such that
as $\xi^2+b\xi+ c$ is an irreducible polynomial over $\GF(q)$.

A set of points $\cO$ is an \emph{elliptic quadric} if there is a
collineation $\tau$ of $\PG(3,q)$ such that ${\cO}^{\tau}=Q^-(3,q)$.
A line, which intersects $\cO$ in exactly one point is called 
\emph{tangent}.
Likewise,
a plane which intersects $\cO$ in a exactly one point
is also called \emph{tangent}.
The following properties are straightforward:
\begin{enumerate}[1.]
\item $\cO$ contains exactly $q^2+1$ points;
\item no three points of $\cO$ are collinear;
\item  a plane meets $\cO$ in either exactly one or
in $q+1$ points;
\item if $P$ is a point of $\cO$, then the lines tangent to
$\cO$ at $P$ are contained in the plane tangent to $\cO$ in
$P$;
\item there is a subgroup of order $q^2+1$ of $\PGL(4,q)$ which acts
transitively on $\cO$.
\end{enumerate}

Let now $\cS$ be a regular spread of $\PG(3,2^e)$.
Then, there is an elliptic quadric $\cO$ such that:
\begin{enumerate}[1.]
\item any line of $\cS$ is tangent to $\cO$;\\
\item the subgroup $S_2$ of order $q^2+1$ of the Singer cycle
  $S$  which stabilises $\cS$, acts regularly on $\cO$.
\end{enumerate}

Let $S_1=\langle \tau \rangle $ be the subgroup of $S$ of order $q+1$
 which fixes all the lines of $\cS$
 and acts transitively on the points of any line of the spread.
 Then, ${\cO}_i={\cO}^{\tau^i}$, for $i=0,1, \dots,q$
 is an elliptic quadric and any line of $\cS$ is tangent to ${\cal
   O}_i$;
 indeed, 
 $\{{\cO}_0, {\cO}_1,\dots,{\cO}_q\}$
 is a partition of point--set of $\PG(3,q)$.

\begin{lm}
\label{lem:1}
For any line $l$ of $\PG(3,q)$ not in $\cS$ there is a unique
quadric $\cO_i$ such that $l$ is tangent to $\cO_i$.
\end{lm}
\begin{proof}
As $q$ is even, there is a symplectic polarity
$\perp$ of $\PG(3,q)$ such that
a line is totally isotropic with respect to $\perp$
if and only if it is tangent to $\cO$.

Call $\pi$ the symplectic polarity induced by ${\cO}_i$ for $i\neq 0$.
Then, a line of $\PG(3,q)$ is totally isotropic simultaneously respect
to $\perp$ and to $\pi$ if and only if $l$ belongs
to $\cS$, see\cite{H}.
Hence, no line is tangent to both ${\cO}$ and ${\cO}_i$.
Thus, ${\cO}_i$ and ${\cO}_j$ do not have any common tangent for $i\neq j$.

Any elliptic quadric ${\cO}_i$
has $(q^2+1)(q+1)$ tangent lines; as
$q^2+1$ of these belong to $\cS$,
we obtain that there are $(q^2+1)q(q+1)$ lines of $\PG(3,q)$
which are tangent to exactly one of the quadrics
${\cO}_0, {\cO}_1,\dots,{\cO}_q$.
Since the number of the lines of $\PG(3,q)$ not in $\cS$ is
$(q^2+1)q(q+1)$, this yields the lemma.
\end{proof}

We are now in the position to state the main theorem
of this section, namely
we provide a geometric construction of a starter set of $\PG(3,q)$.
\begin{thm}
The  lines tangent to $\cO$ through a fixed point $P$
form a starter set of $\PG(3,q)$.
\end{thm}
\begin{proof}
 Let $\alpha$ be the tangent plane to $\cO$ at $P$ and denote by 
$\{m_0,m_1, \dots, m_q\}$ be the $q+1$ tangents to $\cO$ at $P$.
Assume that $m_0$ belongs to  $\cS$.

If there is $\delta\in S$ such that $m_i^{\delta}=m^j$ with $i,j\neq
0$, then $m_0^{\delta}\in{\cS}$, since $\cS$ is fixed by all
the elements of $S$.


Suppose now that $m_0^{\delta}\neq m_0$. 
If it were ${\cO}={\cO}^{\delta}$,  then the line $m_j$ 
would be incident with
both  $P$ and $P^{\delta}$, both in $\cO$; hence,
$m_j$ could not be a tangent line. Thus,
${\cO}\neq {\cO}^{\delta}$.
This would mean that $m_j$ is tangent to both ${\cO}$
and ${\cO}^{\delta}$ --- a contradiction by Lemma \ref{lem:1}. 
It follows that $m_0^{\delta}= m_0$.  
In particular, it follows from this argument
that $\delta$ is in $S_1$.

To conclude the proof, observe that
if $P^{\delta}\neq P$, then $m_j$ is tangent to
${\cO}^{\delta}$ at $P^{\delta}$ and to ${\cO}$ at $P$.
Then, ${\cO}={\cO}^{\delta}$  and $\delta=\mathrm{id}$.
\end{proof}    

\section*{Thanks}
The authors wish to express their gratitude to
prof. G.~Lunardon for several insightful discussions
on the topics here--within investigated.
\bibliographystyle{amsplain}
\bibliography{gls}
\end{document}